\documentclass[letterpaper, 10pt, conference]{ieeeconf}      
\IEEEoverridecommandlockouts                              

\def\QED{\QEDopen}
\usepackage{graphicx}          
\usepackage{microtype}
\usepackage{amsmath}
\usepackage{amssymb}
\usepackage{mathtools}
\usepackage{color,soul}
\usepackage{url}
\usepackage{algpseudocode}
\usepackage{algorithm}
\usepackage{comment}
\usepackage{enumerate}

\usepackage[dvipsnames]{xcolor}
\usepackage{tikz}
\usepackage{pgfplots}
\usepackage{grffile}
\pgfplotsset{compat=newest}
\usetikzlibrary{automata}
\usetikzlibrary{arrows, arrows.meta, decorations.markings, patterns}
\usepgfplotslibrary{patchplots, groupplots}
\usetikzlibrary{shapes, positioning, fit, backgrounds}
\usepgfplotslibrary{fillbetween}

\usetikzlibrary{overlay-beamer-styles}
\usetikzlibrary{calc}

\usepackage{calc}

\definecolor{tudcyan}{RGB}{0,166,214}
\definecolor{tudmagenta}{RGB}{109,23,127}
\definecolor{tudpurple}{RGB}{29,28,115}
\definecolor{tudgraygreen}{RGB}{107,134,137}

\colorlet{lighttudcyan}{tudcyan!20}
\colorlet{lighttudmagenta}{tudmagenta!20}

\newlength{\hatchspread}
\newlength{\hatchthickness}
\newlength{\hatchshift}
\newcommand{\hatchcolor}{}
\tikzset{hatchspread/.code={\setlength{\hatchspread}{#1}},
	hatchthickness/.code={\setlength{\hatchthickness}{#1}},
	hatchshift/.code={\setlength{\hatchshift}{#1}},
	hatchcolor/.code={\renewcommand{\hatchcolor}{#1}}}
\tikzset{hatchspread=7pt,
	hatchthickness=0.5pt,
	hatchshift=0pt,
	hatchcolor=black}
\pgfdeclarepatternformonly[\hatchspread,\hatchthickness,\hatchshift,\hatchcolor]
{custom north west lines}
{\pgfqpoint{\dimexpr-2\hatchthickness}{\dimexpr-2\hatchthickness}}
{\pgfqpoint{\dimexpr\hatchspread+2\hatchthickness}{\dimexpr\hatchspread+2\hatchthickness}}
{\pgfqpoint{\dimexpr\hatchspread}{\dimexpr\hatchspread}}
{
	\pgfsetlinewidth{\hatchthickness}
	\pgfpathmoveto{\pgfqpoint{0pt}{\dimexpr\hatchspread+\hatchshift}}
	\pgfpathlineto{\pgfqpoint{\dimexpr\hatchspread+0.15pt+\hatchshift}{-0.15pt}}
	\ifdim \hatchshift > 0pt
	\pgfpathmoveto{\pgfqpoint{0pt}{\hatchshift}}
	\pgfpathlineto{\pgfqpoint{\dimexpr0.15pt+\hatchshift}{-0.15pt}}
	\fi
	\pgfsetstrokecolor{\hatchcolor}
	\pgfusepath{stroke}
}

\def\centerarc[#1](#2)(#3:#4:#5)
{ \draw[#1] ($(#2)+({#5*cos(#3)},{#5*sin(#3)})$) arc (#3:#4:#5); }

\newtheorem{prob}{Problem}

\newtheorem{defn}{Definition}

\newtheorem{prop}{Proposition}

\newcommand*{\tran}{^{\mkern-1.5mu\mathsf{T}}\!}  

\def\d{\ensuremath{\mathrm{d}}}

\DeclareMathOperator{\Post}{Post}

\DeclareMathOperator{\bigO}{\mathcal{O}}

\DeclareMathOperator{\Val}{Val}
\def\Vcoop{V_{\mathrm{coop}}}
\def\Vadv{V_{\mathrm{adv}}}
\def\VU{V_{\mathrm{U}}}
\def\norm[#1]{\left|#1\right|}
\def\shortnorm[#1]{|#1|}

\def\altsim{\ensuremath{\preceq_{\mathrm{AS}}}}

\def\dummy{}

\def\Xs{\mathcal{X}}
\def\Ys{\mathcal{Y}}
\def\Us{\mathcal{U}}

\def\Vs{\mathcal{V}}
\def\No{\mathbb{N}_{0}}
\def\N{\mathbb{N}}
\def\R{\mathbb{R}}
\def\S{\mathbb{S}}
\def\Q{\mathbb{Q}}

\def\Rs{\mathcal{R}}

\def\Ss{\mathcal{S}}
\def\Ks{\mathcal{K}}
\def\Qs{\mathcal{Q}}

\def\Es{\mathcal{E}}  
\def\Us{\mathcal{U}}  

\def\Ts{\mathcal{T}}

\def\nup{{n_{\mathrm{u}}}}
\def\nx{{n_{\mathrm{x}}}}

\def\xv{\boldsymbol{x}}

\def\xiv{\boldsymbol{\xi}}

\def\zetav{\boldsymbol{\zeta}}

\def\Am{\boldsymbol{A}}
\def\Bm{\boldsymbol{B}}

\def\O{\mathbf{0}}

\def\Mm{\boldsymbol{M}}
\def\Nm{\boldsymbol{N}}

\def\Pm{\boldsymbol{P}}
\def\Qm{\boldsymbol{Q}}

\def\Km{\boldsymbol{K}}

\def\Pl{\Pm_{\mathrm{Lyap}}}
\def\Ql{\Qm_{\mathrm{Lyap}}}

\def\e{\mathrm{e}}

\newcommand{\fakeparagraphnospace}[1]{\vspace{1mm}\noindent\textbf{#1.}}
\newcommand{\fakeparagraph}[1]{\fakeparagraphnospace{#1}}

\title{Self-Triggered Control for Near-Maximal Average Inter-Sample Time}

\author{Gabriel de A.~Gleizer, Khushraj Madnani and Manuel Mazo Jr.
	\thanks{This work is supported by the European Research Council through the SENTIENT project (ERC-2017-STG \#755953).}
	\thanks{G.~de A.~Gleizer, K.~Madnani and M.~Mazo Jr. are with the Delft Center for Systems and Control,
		Delft Technical University, 2628 CD Delft, The Netherlands
		{\tt\small \{g.gleizer, k.n.madnani-1, m.mazo\}@tudelft.nl}}%
}

\begin{document}
	
	\maketitle
	\thispagestyle{empty}
	\pagestyle{empty}
	
	\begin{abstract}
		
		Self-triggered control (STC) is a sample-and-hold control method aimed at reducing communications within networked-control systems; however, existing STC mechanisms often maximize how late the next sample is, and as such they do not provide any sampling optimality in the long-term. In this work, we devise a method to construct self-triggered policies that provide near-maximal  average inter-sample time (AIST) while respecting given control performance constraints. 
		To achieve this, we rely on finite-state abstractions of a reference event-triggered control, in which early triggers are also allowed. These early triggers constitute controllable actions of the abstraction, for which an AIST-maximizing strategy can be computed by solving a mean-payoff game. We provide optimality bounds, and how to further improve them through abstraction refinement techniques.
	\end{abstract}

	\section{INTRODUCTION}
	As networked control systems (NCSs) become the norm in the automation industry, significant attention has been given to sampling and scheduling mechanisms for control. In particular, self-triggered control (STC) \cite{velasco2003self} was proposed as a method to reduce required resource usage. The idea behind STC is that the controller decides the next sampling instant in an aperiodic fashion. In \cite{tabuada2007event}, Tabuada proposed a related method called event-triggered control (ETC), finally establishing a rigorous treatment for stabilization under such aperiodic methods. In ETC, a function between actual state measurements and the past samples is monitored on the sensor side, and data is only transmitted upon an occurrence of a significant event.
	After \cite{tabuada2007event}, much work was dedicated to the design of triggering mechanisms that reduce sampling while guaranteeing stability and performance criteria \cite{wang2008event, girard2015dynamic, heemels2012introduction}. STC was then typically designed by estimating when such events would happen on the controller (e.g., \cite{mazo2010iss}).
	 	
	Among all the aforementioned works, the commonality lies on the philosophy of sampling: it must be as late as possible. Quite often, the \emph{triggering condition} is a surrogate of an underlying Lyapunov function condition, either ensuring its decrease in continuous time  \cite{tabuada2007event}, its boundedness with respect to a decaying function \cite{wang2008event, mazo2010iss, szymanek2019periodic}, or its decrese in average \cite{girard2015dynamic}. The sampling happens as soon as the condition is violated, that is, as late as possible. Posing this sampling strategy as a mechanism to maximize \emph{average inter-sample time} (AIST) --- a natural metric of sampling performance ---, this corresponds to a \emph{greedy} approach for optimization: one maximizes the short-term reward hoping that this brings long-term maximization. It is not surprising that this approach is very rarely optimal, and we have no reason to believe this would be the case with sampling. Aiming at performance improvements, Antunes et al \cite{antunes2012dynamic} addressed the problem of co-designing controllers and sampling strategies that minimize a quadratic control performance while ensuring that the AIST is not smaller than some reference periodic controller. Before the ETC era, the problem of minimizing a weighted sum of quadratic performance and AIST for linear controllers with random disturbances was addressed in \cite{xu2004optimal}, using an elegant dynamic programming approach; unfortunately, this approach lacks the practicality of its STC successors, as the decision to sample or not requires solving an online recursion involving complicated operators. Therefore, we choose to focus on the more practical Lyapunov-based ETC and STC strategies, 
	even though their simplicity comes with a cost: the emerging traffic patterns are often erratic. In fact, even rigorously measuring the AIST of a given ETC implementation is challenging, and only recently a computational method to do this for periodic ETC (PETC, \cite{heemels2013periodic})%
	\footnote{In PETC, events are monitored periodically instead of continuously; this is more realistic in terms of implementation.} %
	has been proposed in \cite{gleizer2021hscc}.
	
	In this work, we address the problem of synthesizing a self-triggered mechanism for state feedback of linear systems that maximizes the AIST, while ensuring the same control performance as a reference PETC. This is attained by considering predicted PETC triggering times as \emph{deadlines}: sampling earlier than these deadlines always ensures equal or better control performance, and done in the right way it can provide long-term benefits in terms of AIST.
	%
	To determine when these long-term rewards are attained, we build on our recent work in \cite{gleizer2021hscc} by using finite-state abstractions \cite{tabuada2009verification} of the reference PETC, and appending weights on the transitions systems, obtaining \emph{weighted automata} \cite{chatterjee2010quantitative}. Finite-state abstractions can simulate all possible behaviors of the concrete systems, while permitting computations that are otherwise intractable or impossible in infinite-state systems.
	By augmenting the finite-state simulations with \emph{controllable} inter-event times up to the PETC deadline, we obtain an alternating simulation relation which also retains quantitative properties. As we will see, solving a \emph{mean-payoff game} on the abstraction gives a strategy that approximately maximizes the AIST; the value of the abstraction game is in fact a lower bound to the optimal one, while computing the (cooperative) maximum AIST of the system gives an upper bound to it. %
	Our approach gives a state-dependent sampling strategy (SDSS) that requires predicting the next $l$ inter-sample times that the reference PETC would generate, $l$ being a chosen parameter. As we will see in a numerical example, even $l=1$ can provide massive improvements to the closed-loop system's AIST by simply using the strategy obtained from the mean-payoff game. The computational cost of this prediction is proportional to $l$, making our approach implementable on hardware with limited capabilities or time-critical applications. In addition, a major benefit of using abstractions is that the methodology is general in the sense that it can be used for more complex control specifications; and while we focus on the control of a single linear system, extensions to multiple controllers sharing a network \cite{kolarijani2015symbolic, gleizer2020scalable} and nonlinear systems \cite{delimpaltadakis2020traffic} are possible using existing abstraction methods.

	\subsection{Notation}
	
	We denote by $\No$ the set of natural numbers including zero, $\N \coloneqq \No \setminus \{0\}$, $\N_{\leq n} \coloneqq \{1,2,...,n\}$, and $\R_+$ the set of non-negative reals. 
	For a set $\Xs\subseteq\Omega$, we denote by $\bar{\Xs}$ its complement $\Omega \setminus \Xs$, and by $|\Xs|$ its cardinality.
	We often use a string notation for sequences, e.g., $\sigma = abc$ reads $\sigma(1) = a, \sigma(2) = b, \sigma(3) = c.$ Powers and concatenations work as expected, e.g., $\sigma^2 = \sigma\sigma = abcabc.$ We denote by $\Xs^+$ (resp.~$\Xs^\omega$) the sets of finite (resp.~infinite) sequences with elements on $\Xs$.
	For a relation $\Rs \subseteq \Xs_a \times \Xs_b$, its inverse is denoted as $\Rs^{-1} = \{(x_b, x_a) \in \Xs_b \times \Xs_a \mid (x_a, x_b) \in \Rs\}$. Finally, we denote by $\pi_\Rs(\Xs_a) \coloneqq \{ x_b \in \Xs_b \mid (x_a, x_b) \in \Rs \text{ for some } x_a \in \Xs_a\}$ the natural projection $\Xs_a$ onto $\Xs_b$.
	
	Given an autonomous system $\dot{\xiv}(t) = f(\xiv(t))$, we say that the origin is globally exponentially stable (GES) if there exist $M < \infty$ and $b > 0$ such that every solution of the system satisfies $|\xiv(t)| \leq M\e^{-bt}|\xiv(0)|$ for every initial state $\xiv(0)$; moreover, we call $b$ a decay rate estimate of the system. When the inital state is not obvious from context, we use $\xiv_{\xv}(t)$ to denote a trajectory from initial state $\xiv(0) = \xv.$
	
	\section{PROBLEM FORMULATION}
	
	Consider a linear plant controlled with sample-and-hold state feedback described by
	\begin{equation}
	\dot{\xiv}(t) = \Am\xiv(t) + \Bm\Km\hat{\xiv}(t), \label{eq:plant}
	\end{equation}
	where $\xiv(t) \in \R^\nx$ is the plant's state, $\hat{\xiv}(t) \in \R^\nx$ is the measurement of the state available to the controller, $\Km\hat{\xiv}(t) \in \R^\nup$ is the control input, $\nx$ and $\nup$ are the state-space and input-space dimensions, respectively, and $\Am, \Bm, \Km$ are matrices of appropriate dimensions. We consider a zero-order hold mechanism: let $t_i \in \R_+, i \in \N_0$ be a sequence of sampling times, with $t_0 = 0$ and $t_{i+1} - t_i > \varepsilon$ for some $\varepsilon > 0$; then $\hat{\xiv}(t) = \xiv(t_i), \forall t \in [t_i, t_{i+1})$. In STC, the controller uses available information at time $t_i$ to determine the next sampling time $t_{i+1}$. Because the system of interest is time-invariant, it suffices to determine the next inter-sample time (IST) $\tau_{i} \coloneqq t_{i+1} - t_i$. Generally speaking, when full-state information is available, STC is a controller composed with a \emph{state-dependent sampling strategy} (SDSS) $s : (\R^\nx)^+ \to \R_+$ that can consider the first $i$ state samples to determine $\tau_{i}$. A \emph{static} state-dependent sampling strategy has the form $s : R^\nx \to \R_+$. The algorithm that dictates $\tau_{i}$ must guarantee given performance specifications while being fast enough to execute in real time, depending on the application.

	Given system \eqref{eq:plant}, and an initial state $\xv \coloneqq \xiv(0),$ an SDSS determines a unique state trajectory $\xiv_{\xv}(t)$, as well as the sequence of IETs $\tau_i$, since $\tau_0 = s(\xv)$ and the following recursion on $i$ holds:
	\begin{equation}\label{eq:sdssrecursion}
		\begin{aligned}
			\xiv_{\xv}(t) &= \Mm(t - t_i)\xiv_{\xv}(t_i), \ \forall t \in [t_i, t_{i+1}] \\
			t_{i+1} &= t_i + \tau_i, \\
			\tau_i &= s(\xiv_{\xv}(t_i), \xiv_{\xv}(t_{i-1}), ... \xiv_{\xv}(t_0)),
		\end{aligned}
	\end{equation}
	where $\Mm(t) \coloneqq \Am_\d(t) + \Bm_\d(t)\Km \coloneqq \e^{\Am t} + \int_0^{t}\e^{\Am\tau}\d\tau \Bm\Km$ is the state transition matrix under the held control input. Thus, for convenience, we denote the unique IET sequence of a given SDSS $s$ from a given initial condition $\xv$ by $\tau_i(\xv; s)$. With this, we can define the average IST (AIST) from $\xv$ as
	$$ \text{AIST}(\xv; s) \coloneqq \liminf_{n\to\infty}\frac{1}{n+1}\sum_{i=0}^{n}\tau_i(\xv; s), $$
	and the smallest AIST (SAIST) across all initial conditions as
	\begin{equation}
		\text{SAIST}(s) \coloneqq \inf_{\xv\in\R^{\nx}}\text{AIST}(\xv; s).
	\end{equation}

	The SAIST of a given SDSS gives its sampling performance, which we are interested in maximizing. %
	
	A standard approach to SDSS design is predicting the time a relevant Lyapunov-based event would happen, as in \cite{mazo2010iss, fiter2012state, heemels2012introduction}. We consider \emph{quadratic triggering conditions}, which encompass the most common ones in the ETC and STC literature, see \cite{heemels2013periodic}. The idea is to design a matrix $\Qm \in \S^{2\nx}$ such that
	\begin{equation}\label{eq:trigcond}
		\begin{bmatrix}\xv \\ \hat\xv \end{bmatrix}\tran
		\!\Qm \begin{bmatrix}\xv \\ \hat\xv\end{bmatrix} \leq 0 \implies \dot{V}(\xv) \leq -aV(\xv)
	\end{equation}
	for some $a > 0$, where $V$ is a Lyapunov function. That is, while the triggering condition (LHS above) is satisfied, some Lyapunov condition is also satisfied, ensuring stability and performance specifications are met. In ETC, this condition is checked continuously and the sample occurs as soon as it is violated. In PETC, this condition is checked every $h$ time units, and further Lyapunov analysis needs to be made to ensure specifications are met, see \cite{heemels2013periodic}. In STC, instead of checking events, the controller predicts them at every $h$ units and typically chooses the last one such that the condition is still met:
	\begin{gather}
		t_{i+1} = \sup\{t \in h\N \mid t > t_i \text{ and } c(t-t_i, \xiv(t), \hat\xiv(t))\}, \label{eq:trigtime}\\
		c(\tau, \xv, \hat\xv) \coloneqq 
		\begin{bmatrix}\xv \\ \hat\xv \end{bmatrix}\tran
		\!\Qm \begin{bmatrix}\xv \\ \hat\xv\end{bmatrix} \leq 0
		\text{ or } \ \tau \leq \bar{\tau},  \label{eq:quadtrig}
	\end{gather}
	where $\bar{\tau}$ is a chosen maximum inter-event time that ensures the search space is finite.%
	\footnote{\label{foot:bartau} Often a maximum $\bar{\tau}$ naturally emerges form a PETC triggering condition. If it ensures GES of the origin, this holds if all eigenvalues of $\Am$ have non-negative real part or, more generally, if the control action $\Bm\Km\xv \neq \O$ for all $\xv$ in the asymptotically stable subspace of $\Am$. In such situations, $\xiv(t)$ does not converge to the origin without intermittent resampling. Still, one may want to set a smaller maximum inter-event time so as to establish a ``heart beat'' of the system.}

	Clearly, the ETC and STC strategies described above are greedy sampling strategies: they maximize the next IST while ensuring the triggering condition is not violated.%
	\footnote{In the PETC case, the triggering condition can be violated for over up to $h$ time units, but one can use a predictive approach as in \cite{szymanek2019periodic} to prevent violations.}
	Therefore, we can use the ETC-generated IST as a state-dependent \emph{deadline} $d : \R^{\nx} \to \Ts$ of an SDSS:
	\begin{equation}\label{eq:deadline}
		d(\xv) \coloneqq \max\{\tau \in \Ts \mid c(\tau, \xiv_{\xv}(\tau), \xv)\},
	\end{equation}
	where $\Ts = \{h, 2h, ..., \bar\tau\}$ is the set of possible inter-event times, and $\bar\tau = hK,$ with $K \in \N$. The question that naturally arises is whether sampling earlier than this deadline can provide long-term benefits in terms of average inter-sample time. This possibility is exactly what we exploit in this work. Hereafter, we shall refer to an SDSS $s^* : (\R^\nx)^+ \to \Ts$ that respects the deadlines in Eq.~\eqref{eq:deadline} as an \emph{early-triggering SDSS}; doing so provides the RHS of Eq.~\ref{eq:trigcond} as a performance guarantee. The overall problem we are interested to solve is as follows:
	
	\begin{prob}\label{prob:ideal}
		Given system \eqref{eq:plant} and a triggering matrix $\Qm \in \S^{2\nx}$, design a SAIST-maximizing early-triggering SDSS $s^* : (\R^\nx)^+ \to \Ts$; i.e., $\text{SAIST}(s^*) \geq \text{SAIST}(s)$ for all early-triggering SDSSs $s$.
	\end{prob}

	\subsection{A simplification of the problem}
	
	Prior to approaching Problem \ref{prob:ideal}, let us see a property of the AIST metric:
	
	\begin{prop}\label{prop:limavgperiodic}
		Let $\tau$ be an infinite sequence that is asymptotically periodic, i.e., its elements can be decomposed as $\tau_i = a_i + b_i$, where $a_i = a_{i+M}$ for all $i\in\N$ and some natural number $M < \infty$, and $b_i$ satisfy $\lim_{i\to\infty} b_i \to 0$. Then, $\liminf_{n\to\infty}\frac{1}{n+1}\sum_{i=0}^{n}\tau_i = \frac{1}{M}\sum_{i=0}^{M-1} a_i.$
	\end{prop}
	\begin{proof}
		We use the superadditivity property of $\liminf$:
		\begin{multline*}
			\liminf_{n\to\infty}\frac{1}{n+1}\sum_{i=0}^{n}\tau_i \\ \geq \liminf_{n\to\infty}\frac{1}{n+1}\sum_{i=0}^{n}a_i + \liminf_{n\to\infty}\frac{1}{n+1}\sum_{i=0}^{n}b_i \\=  \frac{1}{M}\sum_{i=0}^{M-1} a_i + \liminf_{n\to\infty}\frac{1}{n+1}\sum_{i=0}^{n}b_i.	
		\end{multline*}
		Stolz--Cesàro theorem gives that the second term is zero:
		$$ \liminf_{n\to\infty}b_i = 0 \leq \liminf_{n\to\infty}\frac{1}{n+1}\sum_{i=0}^{n}b_i \leq \limsup_{n\to\infty}b_i = 0. $$
		The same arguments can be used with $\limsup$, obtaining $\limsup_{n\to\infty}\frac{1}{n+1}\sum_{i=0}^{n}\tau_i \leq \frac{1}{M}\sum_{i=0}^{M-1} a_i.$ Hence, the sequence $\frac{1}{n+1}\sum_{i=0}^{n}\tau_i$ converges and
		$$ \liminf_{n\to\infty}\frac{1}{n+1}\sum_{i=0}^{n}\tau_i = \lim_{n\to\infty}\frac{1}{n+1}\sum_{i=0}^{n}\tau_i = \frac{1}{M}\sum_{i=0}^{M-1} a_i.$$
	\end{proof}
	
	Proposition \ref{prop:limavgperiodic} shows that, if a sequence of inter-sample times $\tau_i$ converges to a periodic pattern, then AIST is solely a function of the periodic component; in other words, it is \emph{insensitive to transient behavior}, which makes it a fundamentally \emph{long-term,} or \emph{stationary-behavior-dependent} metric. This not only reinforces our argument against a greedy sampling strategy, but also indicates how complicated Problem \ref{prob:ideal} is: it may require an infinite look-ahead capability. Moreover, it is worthless to study how the deadline $d(\xv)$ of Eq.~\eqref{eq:deadline} behaves as $\xv \to 0$ by exploiting GES of the closed-loop system: deadlines are insensitive to the state magnitude.%
	\footnote{It is easy to see that, if $d(\xv) = \tau$ for some $\xv$ in Eq.~\ref{eq:deadline}, then $s(\lambda\xv) = \tau$ for any $\lambda \in \R \setminus \{0\}.$ Similar results have been obtained for other triggering conditions, see \cite{mazo2010iss}.} 
	
	As a final simplification, we focus on \emph{static} SDSSs --- it is possible to show that, given the forgetfulness of the AIST metric, static strategies suffice for optimality; we skip this proof due to space limitations.
	The relaxed problem statement becomes the following:
	\begin{prob}\label{prob:real}
		Given system \eqref{eq:plant} and a triggering matrix $\Qm \in \S^{2\nx}$, design a static early-triggering SDSS $s^* : \R^\nx \to \Ts$ that, given $\epsilon$, satisfies $\text{SAIST}(s^*) \geq \text{SAIST}(s) - \epsilon$ for all early-triggering SDSSs $s$.
	\end{prob}

	\section{Finding an SDSS through abstractions}
	
	To solve Problem \ref{prob:real}, we adapt the abstraction framework of \cite{tabuada2009verification} to include quantitative capabilities; that is, we abstract our system as a \emph{weighted automaton}, in a similar way as we did in \cite{gleizer2021hscc}.
	
	\subsection{Transition systems}
	
	In \cite{tabuada2009verification}, Tabuada gives a generalized notion of transition system. Here, we include a weight function as in \cite{chatterjee2010quantitative} for the quantitative aspect.
	\begin{defn}[Weighted transition system]\label{def:system} 
		A system $\Ss$ is a tuple $(\Xs,\Xs_0,\Us,\Es,\Ys,H,\gamma)$ where:
		\begin{itemize}
			\item $\Xs$ is the set of states,
			\item $\Xs_0 \subseteq \Xs$ is the set of initial states,
			\item $\Us$ is the set of actions,
			\item $\Es \subseteq \Xs \times \Us \times \Xs$ is the set of edges (or transitions),
			\item $\Ys$ is the set of outputs, and
			\item $H: \Xs \to \Ys$ is the output map.
			\item $\gamma: \Es \to \Q$ is the weight function.
		\end{itemize}
	\end{defn}
	A system is said to be finite (infinite) state when the cardinality of $\Xs$ is finite (infinite), and it is said to be just finite if both $\Xs$ and $\Us$ are finite. A transition in $\Es$ is denoted by a triple $(x, u, x')$. We define $\Post_\Ss(x) \coloneqq \{x'\mid \exists u: (x,u,x') \in \Es\}$ as the set of states that can be reached from $x$ in one step. System $\Ss$ is said to be \emph{non-blocking} if $\forall x \in \Xs, \Post_\Ss(x) \neq \emptyset.$ 
	We call $x_0u_0x_1u_1x_2...$ an \emph{infinite internal behavior}, or \emph{run} of $\Ss$ if $x_0 \in \Xs_0$ and $(x_i,u_i,x_{i+1}) \in \Es$ for all $i \in \N$, $y_0y_1...$ its corresponding \emph{infinite external behavior}, or \emph{trace}, if $H(x_i) = y_i$ for all $i \in \N$, and $v_0v_1...$ its \emph{value trace} if $v_i = \gamma(x_i,u_i,x_{i+1})$ for all $i \in \N$. We denote by $\gamma_{\Ss}(r)$ the value trace of $r = x_0u_0x_1...$ (in the case above, $\gamma_{\Ss}(r) = v_0v_1...$), by $\Vs^\omega_x(\Ss)$ the set of all possible value traces of $\Ss$ starting from state $x$, and by $\Vs^\omega(\Ss) \coloneqq \bigcup_{x\in\Xs_0}\Vs^\omega_x(\Ss)$ the set of all possible value traces of $\Ss$. We denote by $\gamma(\Es)$ the set of all possible weight valuations of $\Ss$, which is guaranteed to be finite if $\Ss$ is finite. A state $x$ is called \emph{reachable} if there exists a run $r$ of $\Ss$ containing $x$; for the rest of this paper, we assume every state is reachable. This does not affect generality as one can always remove unreachable states from a transition system without changing its behavior. We denote by $U(x) \coloneqq \{u \mid \exists x' \in \Xs : (x,u,x') \in \Es\}$ the action set of $x \in \Xs$. 
	
	A system $\Ss_b$ is loosely called an \emph{abstraction} of another system $\Ss_a$ if it retains some properties of $\Ss_b$ while being simpler. Typically, we use the term abstraction to denote a finite system $\Ss_a$ that \emph{simulates} or is \emph{alternatingly simulated} by a possibly infinite system $\Ss_b$ (see \cite{tabuada2009verification} for a thorough exposition of these concepts). When one is interested in verification, simulation is used; if one wants to use $\Ss_a$ to find a strategy for $\Ss_b$ that ensures certain properties, the relation of interest is alternating simulation. We use the following definitions, which are amenable to quantitative properties.
	
	\begin{defn}[Weight simulation relation]\label{def:sim}
		Consider two systems $\Ss_a$ and $\Ss_b$ with $\gamma_a(\Es_a)$ = $\gamma_b(\Es_b)$. A relation $\Rs \subseteq \Xs_a \times \Xs_b$ is a \emph{weight simulation relation} from $\Ss_a$ to $\Ss_b$ if the following conditions are satisfied:
		\begin{enumerate}[i)]
			\item for every $x_{a0} \in \Xs_{a0}$, there exists $x_{b0} \in \Xs_{b0}$ with $(x_{a0}, x_{b0}) \in \Rs;$
			\item for every $(x_a, x_b) \in \Rs,$ we have that $(x_a, u_a, x_a') \in \Es_a$ implies the existence of $(x_b, u_b, x_b') \in \Es_b$ satisfying $\gamma_a(x_a, u_a, x_a') = \gamma_b(x_b, u_b, x_b')$ and $(x_a', x_b') \in \Rs.$
		\end{enumerate}
	\end{defn}
	We say that $\Ss_b$ simulates $\Ss_a$ if there is a weight simulation relation from $\Ss_a$ to $\Ss_b$, denoting it by $\Ss_a \preceq \Ss_b$. It is a simple exercise to see that $\Ss_a \preceq \Ss_b$ implies $\Vs^\omega(\Ss_a) \subseteq \Vs^\omega(\Ss_b)$, but the converse is not true.
	\begin{defn}[Alternating weight simulation relation]\label{def:altsim}
		Consider two systems $\Ss_a$ and $\Ss_b$ with $\Us_a \subseteq \Us_b$ and $\gamma_a(\Es_a)$ = $\gamma_b(\Es_b)$. A relation $\Rs \subseteq \Xs_a \times \Xs_b$ is a \emph{alternating weight simulation relation} from $\Ss_a$ to $\Ss_b$ if the following conditions are satisfied:
		\begin{enumerate}[i)]
			\item for every $x_{b0} \in \Xs_{b0}$, there exists $x_{a0} \in \Xs_{a0}$ with $(x_{a0}, x_{b0}) \in \Rs;$
			\item for every $(x_a, x_b) \in \Rs,$ it holds that $U_a(x_a) \subseteq U_b(x_b)$;
			\item for every $(x_a,x_b) \in \Rs, u \in U_a(x_a), x_b' \in \Post_u(x_b),$ there exists $x_a' \in \Post_u(x_a)$ such that $\gamma_a(x_a, u, x_a') = \gamma_b(x_b, u, x_b')$ and $(x_a',x_b') \in R.$
		\end{enumerate}
	\end{defn}
	We say that $\Ss_b$ is an alternating weight simulation of $\Ss_a$, denoting it by $\Ss_a \altsim \Ss_b$. This definition is useful for using strategies for $\Ss_a$ on $\Ss_b$: for any initial state in $\Xs_{b0}$, we can initialize the abstraction $\Ss_a$ with a related state. Then, any action $u$ available at $x_{a0}$ is also available in $\Ss_b$ thanks to condition (ii). Finally, whatever transition $(x_b, u, x_b')$ system $\Ss_b$ takes, we can pick a transition with same weight in $\Ss_a$ that again leads to related states $(x_a', x_b')$, and both systems can continue progressing. The main difference in our definitions from those in \cite{tabuada2009verification}, apart from the weighted aspect, is that we are concerned with an ``output map'' on transitions rather than on states; this is more convenient for our intended application, but not fundamentally different as one could convert a system with weights on transitions to one with weights on states. The following results are useful for our application:
	\begin{prop}\label{prop:altsimgivessim}
		Consider two systems $\Ss_a$ and $\Ss_b$ and an alternating weight simulation relation $\Rs \subseteq \Xs_a \times \Xs_b$ from $\Ss_a$ to $\Ss_b$. If $\forall (x_a, x_b) \in \Rs,  U_a(x_a) = U_b(x_b) $, then $\Rs^{-1}$ is a weight simulation relation from $\Ss_b$ to $\Ss_a$.
	\end{prop} 
	\begin{proof}Condition (i) of Def.~\ref{def:sim} is the flipped version of Def.~\ref{def:altsim} (i). Now assume Def.~\ref{def:sim} (ii) is false for some $(x_b, x_a) \in \Rs^{-1}$; then there is $(x_b, u_b, x_b')$ without a matching transition in $\Ss_a$. However, by assumption, $u_b \in U_a(x_a)$; hence, from Def.~\ref{def:altsim} (iii), every $x_b' \in \Post_{u_b}(x_b)$ has a matching $x_a' \in \Post_{u_b}(x_a)$ s.t.~$(x_a', x_b') \in \Rs$, which is a contradiction.
	\end{proof}
	
	Now, denote by $\Ss|s$ the system $\Ss$ controlled by strategy $s$; the following result is very similar to other results on alternating simulations (e.g.,~\cite{tabuada2009verification}):
	\begin{prop}\label{prop:altsimisgood}
		 Let $\Ss_a$ and $\Ss_b$ be two non-blocking systems such that $\Ss_a \altsim \Ss_b$. Let $\Vs^* \subseteq \Vs^\omega(\Ss_b)$ be a set of weight traces. If there exists a strategy $s: X_b^{+} \to \Us$ such that $\Vs^\omega(\Ss_b|s) \subseteq \Vs^*,$ then there exists a strategy $s' : X_a^{+} \to \Us$ such that $\Vs^\omega(\Ss_a|s') \subseteq \Vs^*.$
	\end{prop}
	\begin{proof}(Sketch) Consider the alternating simulation relation $\Rs \subseteq (X_a,X_b)$. Take any state $x_b \in X_{b0}:$ we can pick a related $x_a$ (Def.~\ref{def:altsim} (i)). Choose $u = s(x_a);$ then $\forall (x_a,u,x_a') \in \Es_a$, $\gamma_a(x_a,u,x_a')$ is the first element of some sequence in $\Vs^*$. From Def.~\ref{def:altsim} (ii), we can use $u$ from $x_b$. From (iii), for any $(x_b,u,x_b') \in \Es_b$ there exists $(x_a,u,x_a') \in \Es_a$ with $(x_a',x_b') \in \Rs$, with matching weights. Since all such $(x_a,u,x_a')$ have valid weights, so does any $(x_b,u,x_b')$. Hence, from any $x_b'$ we can pick a related $x_a'$ and choose $u = s(x_ax_a')$ for the next iteration. Repeating the same arguments recursively concludes the proof.
	\end{proof}
	
	The proof of Prop.~\ref{prop:altsimisgood} gives a way to use a strategy from an abstraction for the concrete system: simply run the abstraction in parallel with the concrete system: from any state $x_b$, pick the winning action $u$ from a related state $x_a$ and move the abstraction forward. If the strategy $s$ is \emph{positional} (static), then the abstraction running in parallel is unnecessary: the alternating simulation relation alone suffices, as long as the relation satisfies $|\{x_a \mid (x_a, x_b) \in \Rs\}| = 1$ for all $x_b \in \Xs_b$, i.e., every concrete state has a single related abstract state. This is the case for quotient systems \cite{tabuada2009verification}, which we will employ here. Hereafter we assume this is true; then, given a static strategy $s_a : \Xs_a \to \Us_a \subseteq \Us_b$, we call $s_b : \Xs_b \to \Us_b$ a \emph{refined strategy to $\Ss_b$}, or simply a \emph{refinement of $s_a$}, by setting $s_b(x_b) = s_a(x_a)$, $x_a$ being the unique state satisfying $(x_a, x_b) \in \Rs.$	
	
	\subsection{Abstractions for optimal average weight}\label{ssec:weight_abstractions}
	
	A transition system according to Def.~\ref{def:system} can be regarded as a game, where player 0 picks the action and player 1 antagonistically picks the transition. The problem of finding a strategy for a finite-state weighted transition system in order to maximize the average weight is known as mean-payoff game, a quantitative game. Quantitative games are games on weighted transition systems where a value function $\Val: \Vs^\omega(\Ss) \to \R$ is defined on runs, and one wants to find a strategy $s$ for player 0 that either
	\begin{enumerate}[i)]
		\item ensures a minimum value $v$, i.e, $\Val(w) \geq v$ for all $w \in \Vs^\omega(\Ss|s)$, or
		\item maximizes the value, i.e., finds $\bar{v}$ such that $\Val(w) \geq \bar{v}$ for all $w \in \Vs^\omega(\Ss|s)$ and that, for any $v > \bar{v}$, there is no strategy $s'$ s.t.~$\Val(w') \geq v$ for all $w' \in \Vs^\omega(\Ss|s')$.
	\end{enumerate}
	In the case of mean-payoff objectives, the value function is
	$$ \Val(v_0v_1...) = \liminf_{n\to\infty}\frac{1}{n+1}\sum_{i=0}^{n}v_i. $$
	There exist positional (i.e., static) optimal strategies for mean-payoff games \cite{ehrenfeucht1979positional},%
	\footnote{Solving the optimal value and strategy of a mean-payoff game has pseudopolynomial complexity, with the best known bound of $\bigO(|\Xs|^2|\Es|W)$ due to \cite{comin2017improved}, where $W$ is the maximum weight. Weights are assumed to be integers for complexity analysis; since we assume weights to be in $\Q$, this can always be done by appropriate normalization.}
	which implies that every $x \in \Xs$ admits an \emph{optimal value} $V(x)$ which is the smallest value obtained from every run starting at $x$ when both players 0 and 1 play optimally. By considering that we cannot control the initial state, i.e., player 1 picks it, we have that the \emph{game value} is $V(\Ss) \coloneqq \bar{v} = \inf\{V(x) \mid x\in\Xs_0\}$.
	Once a strategy $s$ has been decided for a system $\Ss$, $\Ss|s$ becomes a 1-player game, containing only the (antagonistic) environment. In the case of a 1-player game, where the environment chooses initial states and transitions, the value of the game can be taken as a function of its weight behaviors, i.e., $\Vadv(\Ss|s) = \inf\{\Val(w) \mid w \in \Vs^\omega(\Ss|s)\}$. We also define a cooperative value from a given initial condition $\Vcoop^{x}(\Ss) \coloneqq \sup\{\Val(w) \mid w \in \Vs^\omega_x(\Ss)\}$, where player 0 has control of transitions and actions for any given initial state $x$; furthermore, let $\VU(\Ss) \coloneqq \inf_{x\in\Xs_0}\Vcoop^{x}(\Ss)$: essentially, player 1 picks the initial state, but player 0 can choose the run from it. 
	The following result, very similar to what was done in \cite{mazo2011symbolicb} for time-optimal control, gives that a strategy from an abstraction, refined to the concrete system, ensures that at least the abstraction value is attained in the concrete case, while an upper bound can be obtained from $\VU$:
	%
	\begin{prop}\label{prop:bounds}
		Let $\Ss_a \altsim \Ss_b \preceq \Ss_a$. Then, $V(\Ss_b) \leq \VU(\Ss_a).$ Moreover, let $s_a$ be a strategy for $\Ss_a$ such that $\Vadv(\Ss_a|s_a) \geq v$. Then, if $s_b$ is a refinement of $s_a$, it holds that $\Vadv(\Ss_b|s_b) \geq v$.
	\end{prop}
	
	\begin{proof} (Sketch) The first inequality comes from $V(\Ss_b) \leq \VU(\Ss_b) \leq \VU(\Ss_a)$; the first holds by definition of $\VU$, where player 1 is more powerful than in the original game; the second inequality holds from (weight) behavioral inclusion. For the last statement, it can be seen using similar arguments as \cite[Sec.~8.2]{tabuada2009verification} that $\Ss_b|s_b \preceq \Ss_a|s_a$; hence, $\Vs^\omega(\Ss_b|s) \subseteq \Vs^\omega(\Ss_a|s)$, which implies that $\Vadv(\Ss_a|s_a) \geq v \implies \Vadv(\Ss_b|s_b) \geq v.$
	\end{proof}
	
	In light of this result, one can use an abstraction $\Ss_a$ to find a near-optimal strategy for the concrete system $\Ss_b$, and further estimate its optimality gap by computing $\epsilon = \VU(\Ss_a) - \Vadv(\Ss_a|s_a)$.%
	\footnote{The 1-player-game values can be obtained for finite systems using Karp's algorithm \cite{karp1978characterization}, whose complexity, $\bigO(|\Xs||\Es|),$ is much smaller than that for optimal mean-payoff games; for $\VU$, a combination of Karp's algorithm and reachability on graphs can be used, retaining the same complexity.}
	
	\subsection{SDSS design}
	
	Now that we have the relevant notions of abstraction in place, we can proceed to apply them towards an SDSS design for System \eqref{eq:plant}. The first step is to describe the system as an infinite transition system; in this case, the ``control action'' is the (discrete) sampling time, where all sampling times from $h$ until the deadline $d(\xv)$ (Eq.~\eqref{eq:deadline}) are allowed.
	The system is $\Ss \coloneqq (\Xs, \Xs_0, \Us, \Es, \Ys, H, \gamma)$ where
	\begin{multline}\label{eq:original}
		\begin{aligned}
			&\Xs = \Xs_0 = \R^\nx; \\
			&\Us = \Ys = \{1,2,...,K\}; \\
			&\Es = \{(\xv, u, \xv') \mid hu \leq d(\xv) \text{ and } \xv' = \xiv_{\xv}(hu)\}; \\
			&H(\xv) = d(\xv)/h; \\
			&\gamma(\xv,u,\xv') = hu.
		\end{aligned}
	\end{multline}
	Note the special characteristics of this transition system: (i) the input set of a given state is determined by the state's output, which is its deadline; and, (ii) the weight is solely a function of the chosen input. These characteristics simplify the job of find an abstraction that is alternatingly weight simulated by $\Ss$. We can do a standard quotient system \cite{tabuada2009verification} as in \cite{gleizer2020scalable}, where the motivation was to allow early triggers for scheduling of multiple control loops in a single shared network. Better than that, we can start from the $l$-complete models \cite{moor1999supervisory} from \cite{gleizer2021hscc}, which can predict the next $l$ deadlines if the PETC triggering strategy is used, and further augment the abstraction with early trigger actions. Hence, let us recover the relation in \cite{gleizer2021hscc} that allows the construction of the quotient state-space:
	\begin{defn}[Deadline sequence relation \cite{gleizer2021hscc}]\label{def:relation} Given a sequence length $l$, we denote by $\Rs_l \subseteq \Xs \times \Ys^l$ the relation satisfying 
		$(\xv,\dummy{k_1k_2...k_l}) \in \Rs_l$ if and only if
		\begin{subequations}\label{eq:sequence}
			\begin{align}
				\xv &\in \Qs_{k_1}, \label{eq:sequence_k1}\\
				\Mm(hk_1)\xv &\in \Qs_{k_2}, \label{eq:sequence_k2}\\
				\Mm(hk_2)\Mm(hk_1)\xv &\in \Qs_{k_3}, \label{eq:sequence_k3}\\
				& \vdots \nonumber\\
				\Mm(hk_{l-1})...\Mm(hk_1)\xv &\in \Qs_{k_l}, \label{eq:sequence_kl-1}
			\end{align}
		\end{subequations}
		where 
		\begin{equation}\label{eq:setq}
			\begin{gathered}
				\Qs_k \coloneqq \Ks_k \setminus \left(\bigcap_{j=1}^{k-1} \Ks_{j}\right) = \Ks_k \cap \bigcap_{j=1}^{k-1} \bar{\Ks}_{j}, \\
				\Ks_k \coloneqq \begin{cases}
					\{\xv \in \Xs \mid \xv\tran\Nm(hk)\xv > 0\}, & k < K, \\
					\R^\nx, & k = K.
				\end{cases}
			\end{gathered}
		\end{equation}
	\end{defn}
	
	Eq.~\eqref{eq:setq} defines the sets $\Qs_k$, containing the states whose deadline is $hk$. Eq.~\eqref{eq:sequence} simply asserts that a state $\xv \in \R^n$ is related to a state $k_1k_2...k_l$ of the abstraction if the IET sequence that it generates for the next $l$ samples is $hk_1,hk_2,...,hk_l$ when the PETC triggering rule is used, i.e., triggers occur at deadlines. Note that $\Qs_k$ is a conjunction of quadratic inequalities; likewise, denoting $\sigma \coloneqq k_1k_2...k_l$, we can define the set $\Qs_\sigma \coloneqq \{\xv \in \R^{\nx} \mid \xv \text{ satisfies \eqref{eq:sequence}}\}$, which is also given by a conjunction of quadratic inequalities. Then, a transition from some state in $\Qs_\sigma$ to a state in $\Qs_{\sigma'}$ exists if $\exists \xv \in \R^{\nx}$ such that
	\begin{equation}\label{eq:transition}
		\begin{aligned}
			\xv &\in \Qs_\sigma \\
			\Mm(hu)\xv &\in \Qs_{\sigma'},
		\end{aligned}
	\end{equation}
	for some $u$ respecting the deadline, i.e., $u \leq k_1$. We can now define the following abstraction:
	\begin{defn}\label{def:lsim} Given an integer $l \geq 1$, the \emph{$l$-predictive traffic model} of $\Ss$ is the system $\Ss_l \coloneqq \left(\Xs_l, \Xs_l, \Us, \Es_l, \Ys, H_l, \gamma_l\right)$, with 
		\begin{itemize}
			\item $\Xs_l \coloneqq \pi_{\Rs_l}(\Xs)$,
			\item $\Es_l \coloneqq \{(\sigma, u, \sigma') \in \Xs_l \times \Us \times \Xs_l \mid u \leq \sigma(1), \exists \xv \in \R^{\nx} : \text{ Eq.~\eqref{eq:transition} holds}\},$
			\item $H_l(k_1k_2...k_l) \coloneqq k_1.$
			\item $\gamma_l(\sigma, u, \sigma') \coloneqq hu.$
		\end{itemize}
	\end{defn}
	The model above partitions $\R^\nx$ into subsets associated with the next $l$ deadlines times that the related states would generate under the PETC triggering rule. The transition set is determined by verifying Eq.~\eqref{eq:transition}, which is a reachability verification; both the state set and transition set can be determined by solving (non-convex) quadratic inequality satisfaction problems (Eqs.~\eqref{eq:sequence} and \eqref{eq:transition}) which can be done exactly using a satisfiability-modulo-theories (SMT) solver\footnote{For that, the variable is $\xv \in \R^\nx$ and the query is, e.g., $\exists \xv: $Eq.~\eqref{eq:sequence} holds.} such as Z3 \cite{demoura2008z3}, or approximately through convex relaxations as proposed in \cite{gleizer2020scalable}. This system is alternatingly weight simulated by $\Ss$, as desired:
	\begin{prop}\label{prop:mystuffworks}
		The relation $\Rs$ from Def.~\ref{def:relation} is a weight simulation relation from $\Ss$ (Eq.~\eqref{eq:original}) to $\Ss_l$ (Def.~\ref{def:system}), and $\Rs_l^{-1}$ is a alternating weight simulation relation from $\Ss_l$ to $\Ss$.
	\end{prop}
	\begin{proof}(Sketch) The proof of alternating simulation is obtained by checking the conditions of Def~\ref{def:altsim}: (i) is trivially satisfied, and so is (ii) with $U_l(k_1k_2...k_l) = U(\xv) = \{1, 2, ..., k_l\}.$  Condition (iii) is ensured by construction of $\Es_l$ and thanks to the fact that $\gamma(\xv,u,\xv') = \gamma_l(\sigma, u, \sigma') = hu$. The simulation then follows from Prop.~\ref{prop:altsimgivessim}.
	\end{proof}

	\fakeparagraph{Obtaining the strategy and $\epsilon$}
	Now that we have a method to abstract System \eqref{eq:original} into a finite system, we can use the methods from Section \ref{ssec:weight_abstractions} to build a near-optimal SDSS for the abstraction $\Ss_l$, then refine it for $\Ss$. The main question is how to define $l$. Given the results in \cite{gleizer2021hscc}, we suggest the following approach: (i) use \cite{gleizer2021hscc} to compute the exact PETC SAIST, or a close enough under-approximation of it; denote this value by $V(\Ss_l|s_{\mathrm{PETC}})$. Set $l = 1$; Then, (ii) compute $\Ss_l$ (Def.~\ref{def:lsim}) and solve the mean-payoff game for it, obtaining the strategy $s_l$ and the game value estimate $v_l$. After that, with $s_l'$ being a refinement of $s_l$ to $\Ss$, (iii) compute $\VU(\Ss_l)$ and verify, using a similar approach to \cite{gleizer2021hscc}, (a sufficiently close under-approximation of) $\Vadv(\Ss|s_l')$
	\footnote{Even though \cite{gleizer2021hscc} was proposed for the PETC strategy, the same approach can be used for any fixed sampling strategy, as its essential feature is verifying cycles in the concrete system.%
	}
	Finally, (iv) compute $\epsilon = \VU(\Ss_l) - V(\Ss_l)$ and (v) if the improvement $\Vadv(\Ss|s_l') - V(\Ss_l|s_{\mathrm{PETC}})$ is large enough, $\epsilon$ is small enough, or $l$ is too large,%
	\footnote{The value of $l$ is proportional to the complexity of the online part of the algorithm, in which the controller must predict the next $l$ deadlines of the current state $\xv$ under PETC. This involves simulating the PETC forward $l$ steps, which takes at most $lK$ operations of quadratic inequalities. Hence, a maximum $l$ may be needed due to computational constraints in the operation.}	
	stop; otherwise, increment $l$ and redo steps (ii) to (v).

	\section{Numerical example}\label{sec:num}
	
	Consider a plant and controller of the form \eqref{eq:plant} from \cite{tabuada2007event}
	\begin{equation*} \Am = \begin{bmatrix}0 & 1 \\ -2 & 3\end{bmatrix}, \ \Bm = \begin{bmatrix}0 \\ 1\end{bmatrix}, \ \Km = \begin{bmatrix}1 & -4\end{bmatrix}
	\end{equation*}
    with the predictive Lyapunov-based triggering condition \cite{gleizer2020towards, szymanek2019periodic} of the form
	$ V(\zetav(t)) > -\rho \zetav(t)\tran\Ql\zetav(t), $
	where $\zetav(t) \coloneqq \Am_\d(h)\xiv(t) + \Bm_\d(h)\Km\hat{\xiv}(t)$ is the next-sample prediction of the state, $V(\xv) = \xv\tran\Pl\xv$, and $0 < \rho < 1$ is the triggering parameter. The Lyapunov matrices were taken from \cite{tabuada2007event} as $\Pl = \begin{bsmallmatrix}1 & 0.25 \\ 0.25 & 1\end{bsmallmatrix}, \ \Ql = \begin{bsmallmatrix}0.5 & 0.25 \\ 0.25 & 1.5\end{bsmallmatrix}$, and we set $h = 0.1$ and $K = 20$, but the largest $K$ according to footnote \ref{foot:bartau} is equal to 11. A minimum-average-cycle-equivalent simulation \cite{gleizer2021hscc} is found for the PETC, giving a SAIST of approx.~0.233 (obtained with $l=8$).  %
	The strategy obtained with $l=1$ already gives massive improvements: it increases SAIST to $\Vadv(\Ss|s_1') = 0.5,$ which is more than twice the PETC's SAIST. Essentially, the obtained strategy simply limits $s(\xv)$ to $5$ for any $k \geq 5$ satisfying $(\xv, k) \in \Rs_1$. This surprisingly shows that limiting the maximum IET can actually have long-term benefits. With $l=2$, the SAIST is improved further to 0.6; a simulation comparing this strategy and PETC is depicted in Fig.~\ref{fig:dead}. One can see that only around $t=2.5$ the SDSS samples before the deadline for the first time; doing so prevents the bursts of IST equal to $0.1$ that happen recurrently with the PETC. The difference in (simulated) running averages between PETC and our SDSS is displayed in Fig.~\ref{fig:runavg}. The SAIST for $l=3$ does not improve, and the upper bounds are $\VU(\Ss_1) = \VU(\Ss_2) = \VU(\Ss_3) = hK = 1.1$.
	\begin{figure}
		\begin{center}
			\input{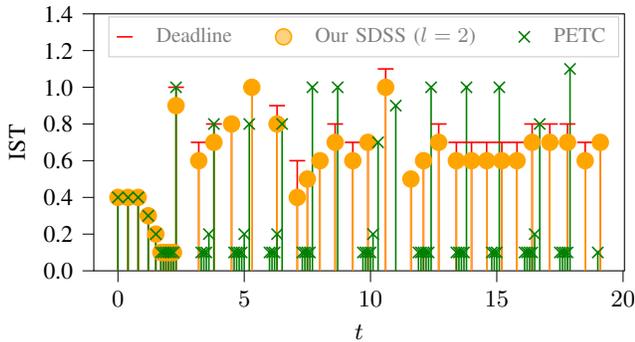}
			\vspace{-2em}
			\caption{\label{fig:dead} {Comparison between simulated traces of our SDSS ($l=2$) and of the PETC, both with the same initial state.}}
			\vspace{-1.5em}
		\end{center}
	\end{figure}
	\begin{figure}
		\begin{center}
			\vspace{0.7em}
			\input{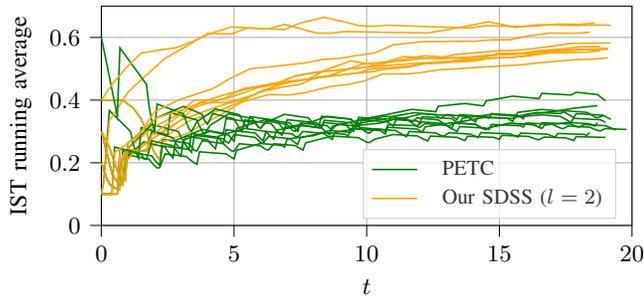}
			\vspace{-1.5em}
			\caption{\label{fig:runavg} {Running average of the ISTs generated from 10 different initial conditions under PETC and our near-optimal SDSS using $l = 2$.} }
			\vspace{-1.5em}
		\end{center}
	\end{figure}
	
	\section{DISCUSSION AND CONCLUSIONS}
	
	In this paper we have presented an abstraction-based approach to build aperiodic sampling strategies for LTI systems in order to maximize their average inter-sample time. For this we  rely on the properties of PETC strategies, which ensure stability and performance of the closed loop whenever their ``deadlines'' are respected. This makes our abstraction inherently safe from a control perspective, but one could relax this condition by allowing ``late'' samplings as long as the abstraction can incorporate some information about the control performance. Multiple ideas can be explored in this direction: e.g, further partioning the state-space with Lyapunov level sets and add reachability and safety specifications; or adding a Lyapunov-based cost to the transitions of the abstraction. In the latter case, one can either have a safety or reachability objective to the Lyapunov function (e.g., maximize AIST subject to $V(\xiv(t)) \leq r$ for all $t > T$), or play a multi-objective quantitative game (e.g., find the Pareto-optimal set of strategies that maximize AIST and control performance).%
	\footnote{Multi-objective mean-payoff games have been addressed in, e.g., \cite{velner2015multi}.}
	
	As seen in our numerical example, the gap $\epsilon$ is still very large after a few refinements. This can be related to the way we do our abstractions: we build a quotient state set based on the deadlines instead of the inputs, which is not how the standard bisimulation algorithm operates. Doing so has benefits for the complexity of the online operations, as to find the related abstract state one needs to simulate an autonomous system (the reference PETC) instead of an exploration involving all possible sampling times. However, this approach cannot (generally) eliminate the nondeterminism associated with sampling earlier than the deadline. Hence, to obtain $V(\Ss_l) \to V(\Ss)$ as $l \to \infty$, 
	refining based on the input set may be necessary; this is a subject of current investigation.
	
	Finally, as with most abstraction-based approaches, our method suffers from the curse of dimensionality. Even though we shift complexity to an offline phase, computing the abstractions $\Ss_l$ can easily be intractable as $\nx$ gets large enough. Approximate methods to solve the satisfiability problems involved in building $\Ss_l$ (e.g., as in \cite{gleizer2020scalable}) are subject of current investigation.

	\bibliographystyle{ieeetr} 
	\bibliography{mybib} 
		
\end{document}